\documentclass[11pt]{article}

\usepackage{a4}
\usepackage[top=30truemm,bottom=30truemm]{geometry}

\usepackage{authblk}

\usepackage{amsthm}
\usepackage{amssymb}
\usepackage{graphicx}

\newtheorem{theorem}{Theorem}
\newtheorem{lemma}{Lemma}
\newtheorem{corollary}{Corollary}

\newtheorem{definition}{Definition}

\newcommand{\polylog}{\mathrm{polylog}}

\title{A faster reduction of the dynamic time warping distance \\ to the longest increasing subsequence length}

\author{Yoshifumi~Sakai$^1$}
\author{Shunsuke~Inenaga$^{2,3}$}

\affil{
  \textit{$^1$ Graduate School of Agricultural Science, Tohoku University, Sendai, Japan} \\
  \texttt{yoshifumi.sakai.c7@tohoku.ac.jp} \\ 
  \textit{$^2$ Department of Informatics, Kyushu University, Fukuoka, Japan}\\
  \texttt{inenaga@inf.kyushu-u.ac.jp} \\ 
  \textit{$^3$ PRESTO, Japan Science and Technology Agency, Kawaguchi, Japan}\\
}

\date{}

\begin{document}
\maketitle

\begin{abstract}
The similarity between a pair of time series, i.e., sequences of indexed values in time order, is often estimated by the dynamic time warping (DTW) distance, instead of any in the well-studied family of measures including the longest common subsequence (LCS) length and the edit distance.
Although it may seem as if the DTW and the LCS(-like) measures are essentially different, we reveal that the DTW distance can be represented by the longest increasing subsequence (LIS) length of a sequence of integers, which is the LCS length between the integer sequence and itself sorted.
For a given pair of time series of length $n$ such that the dissimilarity between any elements is an integer between zero and $c$, we propose an integer sequence that represents any substring-substring DTW distance as its band-substring LIS length.
The length of the produced integer sequence is $O(c n^2)$, which can be translated to $O(n^2)$ for constant dissimilarity functions.
To demonstrate that techniques developed under the LCS(-like) measures are directly applicable to analysis of time series via our reduction of DTW to LIS, we present time-efficient algorithms for DTW-related problems utilizing the semi-local sequence comparison technique developed for LCS-related problems. \\

\noindent \textbf{keywords:} string algorithms; dynamic time warping distance; longest increasing subsequence; semi-local sequence comparison.
\end{abstract}

\maketitle

\section{Introduction}\label{sec intro}

A time series is a sequence of discrete objects which are indexed in time order.
Due to the recent developments of sensing technologies and semi-automated M2M communications, a vast amount of time series data has been rapidly produced in industrial, financial, medical, and scientific domains.

The most fundamental task in time series data analytics is to compare time series sequences, and to extract their similarities.
The \emph{dynamic time warping} (\emph{DTW}) distance is a fundamental method to compute a similarity between two time series that may vary in speed.
It is essentially composed of computing an optimal one-to-many alignment of two time series.
Considering one-to-many mappings allows for dynamic shifts of time points, and it has made DTW one of the most successful algorithms in all areas of algorithms.
Indeed, not only is DTW widely utilized in time series data analysis~\cite{Muller}, but also DTW has been extended to a wide range of other applications including image processing~\cite{RM}, hand writing matching~\cite{TSW}, sign language recognition~\cite{JCA}, music retrieval~\cite{JL}, robotics~\cite{JK,JXWL}, trajectory data analysis~\cite{VG,IMA}, speech recognition~\cite{SC,MBE}, and many others.

Consider two time series sequences $A$ and $B$.
For the time being, let us assume for simplicity that $|A| = |B| = n$.
There is a fundamental dynamic programming algorithm that computes the DTW distance, together with an alignment achieving the distance, between $A$ and $B$ in $O(|A||B|) = O(n^2)$ time and space~\cite{SC}.
While it is possible to reduce the space-requirement of this dynamic programming method to $O(n)$ by applying Hirschberg's divide-and-conquer algorithm~\cite{Hirschberg}, no strongly sub-quadratic time algorithm for computing the DTW distance is known.
This is supported by the conditional lower bound such that, unless the Strong Exponential Time Hypothesis (SETH) is false, there is no $O(n^{2-\epsilon})$-time algorithm for any $\epsilon > 0$ that computes the exact value of the DTW distance of two given sequences over 5-letter alphabets~\cite{ABW,BK}.
Later, the same conditional lower bound was shown for 3-letter alphabets in the case where the cost function $d$ satisfies $d(a, b) = 1$ for any pair $a, b$ of letters~\cite{Kuszmaul}.

On the practical side, a number of fast heuristic algorithms for DTW have been proposed by the database community (see~\cite{WMDTSK} for a survey).
These algorithms typically output approximated values for the DTW distance which in many cases suffice for practical purposes, but, lack theoretical guarantees.

Unlike other sequence comparison measures such as longest common subsequences (LCS) and edit distance, DTW is not a one-to-one/zero alignment.
In addition, the underlying grid graph for DTW is vertex-weighted, while those for LCS and edit distance are edge-weighted.
Despite these different natures of DTW from those of LCS and edit distance, interestingly, computing LCS and weighted edit distance of two sequences of length $n$ can be reduced to computing DTW of two sequences of length $O(n)$~\cite{ABW,Kuszmaul}.
Thus, computing the exact DTW distance is at least as hard as for computing LCS and (weighted) edit distance.
On the other hand, it is not known whether computing DTW can be reduced to computing LCS or (weighted) edit distance.
These are most probably why finding an efficient algorithm for the exact DTW distance is rather challenging, and quite intriguing.
Indeed, the first weakly sub-quadratic time algorithm for the DTW distance, which runs in $O(n^2 \log \log \log n / \log \log n)$ time, was only recently discovered~\cite{GS}, after 40 years from the seminal paper~\cite{SC}.
More recently, the running time has been improved to $O(n^2 / \log \log n)$~\cite{GS2}.

A few DTW algorithms whose running times depend on other parameters are also known:
Hwang and Gelfand~\cite{HG1} showed how to compute the DTW distance in $O((s + t)n)$ time, where $s$ and $t$ denote the number of non-zero values in $A$ and $B$, respectively.
For the case where the minimum non-zero distance is one, Kuszmaul~\cite{Kuszmaul} proposed an algorithm for computing the DTW distance in $O(n u)$ time, where $u$ denotes the DTW distance between $A$ and $B$.
Very recently, Froese et al.~\cite{FJRW} presented a run-length-encoding (RLE) based algorithm which computes the DTW distance in $O(kn + \ell m)$ time, where $m = |A|$, $n = |B|$, and $k$ and $\ell$ are respectively the RLE sizes of $A$ and $B$.
In the case where $k \in O(\sqrt{m})$ and $\ell \in O(\sqrt{n})$, their algorithm runs in $O(k^2 \ell + \ell^2 k)$ time.

When $A$ and $B$ are both binary sequences, it is known that the DTW distance can be computed in $O(n^{1.87})$ time~\cite{ABW}.
There are other DTW algorithms for binary sequences, running in $O(s t)$ time~\cite{MCAHM,HG2}, or in $O(k\ell)$ time~\cite{DM}.
Very recently, a surprising $O(n)$-time solution has been proposed for computing the DTW distance of binary sequences~\cite{Kuszmaul2}.
In the same paper~\cite{Kuszmaul2}, an $O((k + \ell) \log (k + \ell))$-time solution was also proposed for the DTW distance of binary sequences.

\subsection{Reducing DTW to LIS}

In a previous version of this paper~\cite{SI} we presented a new approach for computing the DTW distance, based on a reduction to the \emph{longest increasing subsequence} (\emph{LIS}) problem.

Of many variants of DTW distance between $A$ and $B$ with respect to the dissimilarity between the value $a$ at any position $i$ in $A$ and the value $b$ at any position $j$ in $B$, $d(a, b) = |a - b|$ or $d(a, b) = (a - b)^2$ seems most typically used.
We here adopt a general dissimilar function $d_{A, B}(a, b)$, instead of such a specific function, but use a standard convention that the value $d_{A, B}(a, b)$ is rounded to an integer between $0$ and some positive integer $c$.

In~\cite{SI} the authors presented how the problem of computing the DTW distance between $A$ and $B$ can be reduced to computing the LIS of a sequence of $O(c^2 mn)$ integers, in $O(c^2 mn)$ time and space, where $m = |A|$ and $n = |B|$.
The merit of this method is that it allows us to perform efficient \emph{semi-local sequence comparison}~\cite{Tis} between contiguous subsequences of $A$ and $B$ based on the DTW metric, which further permits us several sophisticated comparisons of the two input sequences in $O(c^2 n^2 \polylog(n))$ time assuming $m = n$.
(More detailed description of these problems and their solutions will be given later in this paper).
We remark that a direct application of the standard DP requires $O(n^3)$ time for these comparisons.

For long sequences with large $m$ and $n$,
the value of $c$ is often negligibly small and thus can be regarded as a constant in many cases.
In particular, $c = 1$ always holds for binary time series such as spike trains
and sensor event sequences.
In these cases, the DTW distance is represented by the LIS length of an integer sequence of length $O(mn)$, or $O(n^2)$ when $m = n$.

On the other hand, for some applications where $c$ is a relatively large constant,
the $c^2$ overhead in the $O(c^2 mn)$ complexity may become problematic.
Thus, whether one can reduce this $c^2$ factor has been an intriguing and important question.

\subsection{Our contribution}

In the first part of this paper, we present a new, more efficient reduction technique of DTW to LIS.
The main part of our new reduction is to somehow redefine the DTW distance between any contiguous subsequences of $A$ and $B$ as the maximum possible weight of a restricted increasing subsequences of a sequence of $O(mn)$ integers each weighted by an integer between $0$ and $c$, where $m$ and $n$ are the lengths of $A$ and $B$, respectively.
This weighted integer sequence is transformed into a sequence of $O(cmn)$ unweighted integers, which essentially has the same property as the weighted sequence, in a straightforward way.
This permits us to speed-up the reduction by a factor of $c$,
compared to the previous approach~\cite{SI}.

For simplicity, assume $m \leq n$.
While our new $O(cn^2)$ reduction is still less time- and space-efficiently than the classical $O(n^2)$ dynamic programming method, our new reduction also permits us efficient \emph{semi-local sequence comparison} with DTW as in~\cite{SI}.
The semi-local sequence comparison problem was first considered by Tiskin~\cite{Tis} with LCS.
The task for our case is to preprocess input strings $A$ and $B$ to construct a data structure supporting $O(n^2)$ queries of the DTW distance of any pair of either a prefix of one of $A$ and $B$ and a suffix of the other, or a contiguous subsequence of one and the entire sequence of the other.
There are two na\"ive solutions for this problem:
The first na\"ive solution \textsf{NS1} is to store the input strings $A$ and $B$ with $O(n)$ space
and to apply the dynamic programming method upon query using $O(n^2)$ time each.
The second na\"ive solution \textsf{NS2} is to precompute a lookup table of $O(n^2)$ space
which explicitly stores all the answers for all possible $O(n^2)$ queries,
allowing for answering each query in $O(1)$ time.
Suppose that $c$ can be treated as a constant independent of $n$ (and of $m$).
Compared to \textsf{NS1},
our $O(n^2)$-space data structure supporting $O(\polylog(n))$-time queries
achieve exponential speed-up for answering semi-local DTW distance queries,
at the sacrifice of quadratic space usage.
Compared to \textsf{NS2}, our $O(\polylog(n))$-time queries are slower than $O(1)$-time queries of \textsf{NS2}.
However, a na\"ive application of the dynamic programming method needs $O(n^3)$ preprocessing time to compute the lookup table of \textsf{NS2}, while our data structure can be built in faster $O(n^2 \polylog(n))$ time.
To summarize, our method leads to a non-trivial time-space trade-off for this semi-local sequence comparison problem with DTW.

We further emphasize that, despite the different nature of DTW from that of LCS or edit distance noted previously, our reduction of DTW to LIS allows us to apply Tiskin's semi-local sequence comparison technique, originally developed for LCS-related problems, directly to DTW-related problems.
As such applications, we present time-efficient algorithms for the \emph{circular DTW distance}, \emph{square root DTW distance}, and \emph{periodic DTW distance} problems, which can arise in time series data analysis.

\section{Preliminaries}\label{sec pre}

For any sequences $S$ and $T$, let $S \circ T$ denote the concatenation of $S$ followed by $T$.
For any sequence $S$, we use $|S|$ to denote the length of $S$ and $S[i]$ with $1 \leq i \leq |S|$ to denote the $i$th element of $S$, so that $S = S[1] \circ S[2] \circ \cdots \circ S[|S|]$.
A subsequence of a sequence $S$ is obtained from $S$ by deleting zero or more elements at any position not necessarily contiguous, i.e., $S[i_1] \circ S[i_2] \circ \cdots \circ S[i_\ell]$ with $1 \leq \ell \leq |S|$ and $1 \leq i_1 < i_2 < \cdots < i_\ell \leq |S|$.
Any subsequence $S[i_\vdash] \circ S[i_\vdash + 1] \circ \cdots \circ S[i_\dashv]$ with $1 \leq i_\vdash \leq i_\dashv \leq |S|$ is called contiguous and denoted by $S[i_\vdash : i_\dashv]$.
A prefix (resp. suffix) of $S$ is a contiguous subsequence $S[i_\vdash : i_\dashv]$ with $i_\vdash = 1$ (resp. $i_\dashv = |S|$).

\begin{figure}[t]
\centering
\includegraphics[width=9cm]{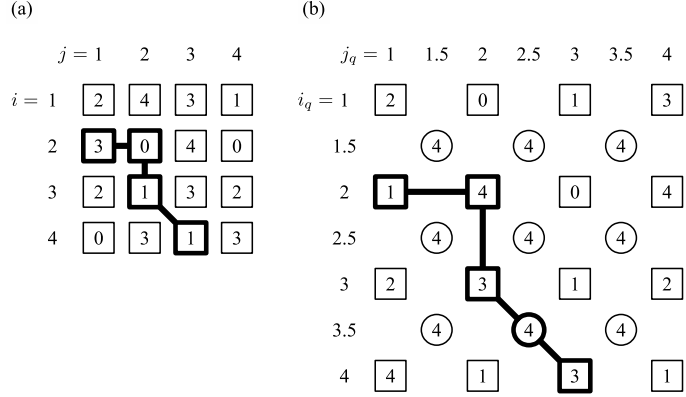}
\caption{
(a) An example of the table of the dissimilarities $d_{A, B}(i, j)$ with $|A| = |B| = 4$ and an alignment $P = (2, 1) \circ (2, 2) \circ (3, 2) \circ (4, 3)$ achieving the DTW distance between $A[2 : 4]$ and $B[1 : 3]$, which is indicated by a polygonal line; (b) The weights $w(q)$ for elements $q$ in sequence $R_{A, B}$ for the same dissimilarity function $d_{A, B}$ as (a) with $c = 4$, where each boxed integer represents $w(r(i, j)) \ (= c - d_{A, B}(i, j))$, each encircled integer represents $w(\tilde{r}(i, j)) \ (= c)$, and sequence $\hat{P} = r(2, 1) \circ r(2, 2) \circ r(3, 2) \circ \tilde{r}(4, 3) \circ r(4, 3)$ for the same $P$ as (a) is indicated by a polygonal line.
}
\label{fig d}
\end{figure}

A \emph{time series} is a nonempty finite sequence.
For any pair of time series $A$ and $B$ and any pair of their elements $A[i]$ and $B[j]$, let $d_{A, B}(i, j)$ denote a nonnegative integer that represents the \emph{dissimilarity} of $A[i]$ and $B[j]$.
An \emph{alignment} of $A[i_\vdash : i_\dashv]$ and $B[j_\vdash : j_\dashv]$ is a sequence $P = (i_1, j_1) \circ (i_2, j_2) \circ \cdots \circ (i_{|P|}, j_{|P|})$ of index pairs such that $(i_1, j_1) = (i_\vdash, j_\vdash)$, $(i_k, j_k)$ with $2 \leq k \leq |P|$ is one of $(i_{k - 1} + 1, j_{k - 1} + 1)$, $(i_{k - 1} + 1, j_{k - 1})$, or $(i_{k - 1}, j_{k - 1} + 1)$, and $(i_{|P|}, j_{|P|}) = (i_\dashv, j_\dashv)$.
The \emph{discrepancy} $d_{A, B}(P)$ of this alignment $P$ is defined as the sum of the dissimilarity $d_{A, B}(i_k, j_k)$ over all indices $k$ with $1 \leq k \leq |P|$.
The \emph{dynamic time warping (DTW) distance} between $A[i_\vdash : i_\dashv]$ and $B[j_\vdash : j_\dashv]$ is defined as the minimum of the discrepancy $d_{A, B}(P)$ over all alignments $P$ of $A[i_\vdash : i_\dashv]$ and $B[j_\vdash : j_\dashv]$.
See Figure~\ref{fig d}(a), in which a concrete example of the dissimilar function $d_{A, B}$ and an alignment $P$ achieving the DTW distance are presented.

For any sequence $S$ of integers, a subsequence of $S$ is \emph{increasing}, if any element of the subsequence other than the last one is less than the succeeding element.
Any increasing subsequence $S'$ of $S$ is $[h_\vdash : h_\dashv]$-\emph{banded}, if $S'$ consists only of integers in the range from $h_\vdash$ to $h_\dashv$.
Any increasing subsequence $T$ of $S$ is \emph{maximal}, if $T$ is the only increasing subsequence of $S$ that has $T$ itself as its subsequence.
The \emph{longest increasing subsequence (LIS) length} of $S$ is the maximum of $|T|$ over all increasing subsequences $T$ of $S$.
Any increasing subsequence of $S$ that achieves the LIS length of $S$ is called an \emph{LIS} of $S$.
The LIS problem introduced above can naturally be generalized for the case where each integer $s$ in $S$ is weighted by a non-negative integer $w(s)$ as follows.
The \emph{heaviest increasing subsequence (HIS) weight} of $S$ is defined as the maximum of $w(T)$ over all increasing subsequence $T$ of $S$, where $w(T)$ denotes the sum of $w(s)$ over all integers $s$ in $T$.
Any increasing subsequence of $S$ that achieves the HIS weight of $S$ is called an \emph{HIS} of $S$.
Note that there is at least an HIS that is maximal, due to the non-negativity of $w(s)$.

\section{Reduction}\label{sec reduction}

Let $A$ and $B$ be arbitrary time series, and let $c$ be the maximum of $d_{A, B}(i, j)$ over all index pairs $(i, j)$ with $1 \leq i \leq |A|$ and $1 \leq j \leq |B|$.
This section designs the DTW distance sequence $S_{A, B}$ for time series $A$ and $B$, which is a sequence of $O(c|A||B|)$ integers that can be used to determine the DTW distance between any pair of $A[i_\vdash : i_\dashv]$ and $B[j_\vdash : j_\dashv]$ as the $[h_\vdash : h_\dashv]$-banded LIS length of $S_{A, B}[g_\vdash : g_\dashv]$ for certain indices $g_\vdash$, $g_\dashv$, $h_\vdash$, and $h_\dashv$.

To define the DTW distance sequence $S_{A,B}$ for $A$ and $B$, we first reduce the DTW distance problem to the HIS problem, and then to the LIS problem. The outline is as follows.

The reduction of the DTW distance problem to the HIS problem is done by introducing a sequence $R_{A, B}$ of $|A||B| + (|A| - 1)(|B| - 1)$ weighted integers, which are  
\begin{itemize}
\item integers $r(i, j)$ for all index pairs $(i, j)$ with $1 \leq i \leq |A|$ and $1 \leq j \leq |B|$, each weighted by $c - d_{A, B}(i, j)$, and 
\item integers $\tilde{r}(i, j)$ for all index pairs $(i, j)$ with $2 \leq i \leq |A|$ and $2 \leq j \leq |B|$, each weighted by $c$.
\end{itemize}
For any element $q$ in $R_{A, B}$, let $i_q = i$ and $j_q = j$, if $q = r(i, j)$, and for convenience, let $i_q = i - 0.5$ and $j_q = j - 0.5$, if $q = \tilde{r}(i, j)$.
For any alignment $P$ of $A[i_\vdash : i_\dashv]$ and $B[j_\vdash : j_\dashv]$, let $\hat{P}$ denote the sequence obtained from $P$ by replacing each element $(i_k, j_k)$ with $2 \leq k \leq |P|$ and $(i_k, j_k) = (i_{k - 1} + 1, j_{k - 1} + 1)$ by $\tilde{r}(i, j) \circ r(i, j)$ and each other element $(i_k, j_k)$ by $r(i, j)$.
Hence, for any consecutive elements $q$ and $q'$ in $\hat{P}$, all of $i_q \leq i_{q'}$, $j_q \leq j_{q'}$, and $(i_{q'} - i_q) + (j_{q'} - j_q) = 1$ hold.
From this observation, the length of $\hat{P}$ is $|A[i_\vdash : i_\dashv]| + |B[j_\vdash : j_\dashv]| + 1$, and hence the sum of the weight $w(q)$ over all integers $q$ in $\hat{P}$ is equal to $c(|A[i_\vdash : i_\dashv]| + |B[j_\vdash : j_\dashv]| + 1) - d_{A, B}(P)$.
See Figure~\ref{fig d}(b), in which the weights $w(q)$ for all elements $q$ in $R_{A, B}$ for the same dissimilarity function $d_{A, B}$ as (a) and sequence $\hat{P}$ for the same alignment $P$ as (a) are presented.
We will carefully define $R_{A, B}$ by specifying an integer value and occurrence position for each element $q$ in $R_{A, B}$ so that any alignment of $A[i_\vdash : i_\dashv]$ and $B[j_\vdash : j_\dashv]$ represents an $[r(i_\vdash, j_\vdash) : r(i_\dashv, j_\dashv)]$-banded maximal increasing subsequence of $R_{A, B}[f_\vdash : f_\dashv]$, and vice versa.
This immediately implies that $c(|A[i_\vdash : i_\dashv]| + |B[j_\vdash : j_\dashv]| + 1)$ minus the $[r(i_\vdash, j_\vdash) : r(i_\dashv, j_\dashv)]$-banded HIS weight of $R_{A, B}[f_\vdash : f_\dashv]$ represents the DTW distance between $A[i_\vdash : i_\dashv]$ and $B[j_\vdash : j_\dashv]$.

The DTW distance sequence $S_{A, B}$, which consists of unweighted integers, will be defined by blowing up each integer in $R_{A, B}$ based on its weight in a straightforward manner.

\subsection{Reduction of the DTW distance problem to the HIS problem}\label{sec R}

\begin{figure}[t]
\centering
\includegraphics[width=12cm]{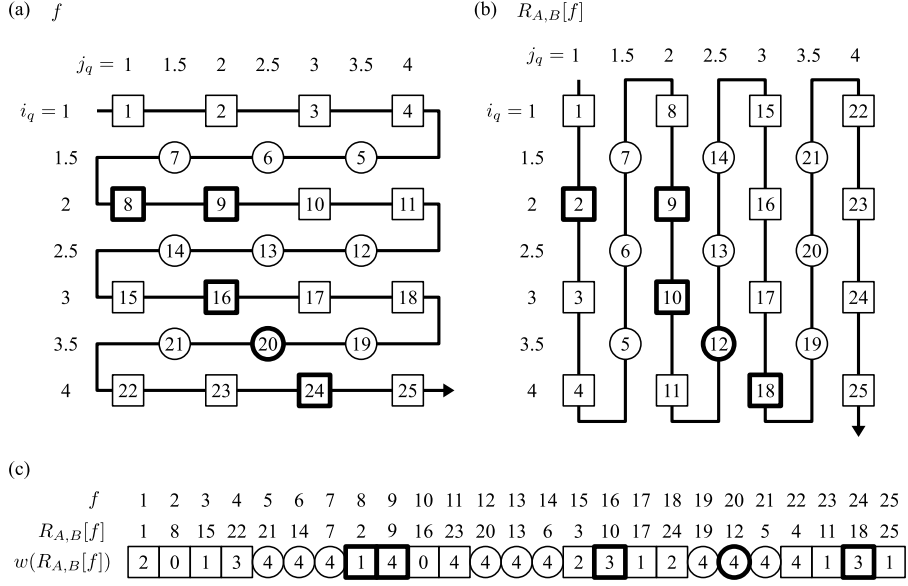}
\caption{
(a) The indices $f(q)$ for elements $q$ in $R_{A, B}$ for the same dissimilarity function as Figure~\ref{fig d}(a), to which indices $1,2,\dots,|R_{A, B}|$ are assigned in the order indicated by a winding arrow, where each boxed index represents $f(r(i, j))$ and each encircled index represents $f(\tilde{r}(i, j))$; (b) The integers of elements $q \ (= R_{A, B}[f(q)])$ in $R_{A, B}$, to which integers $1,2,\dots,|R_{A, B}|$ are assigned in the order indicated by a winding arrow, where each boxed integer represents $r(i, j)$ and each encircled integer represents $\tilde{r}(i, j)$; (c) Sequence $R_{A, B}$, in which the alignment $P$ in Figure~\ref{fig d}(a) corresponds to the $[1 : 18]$-banded increasing subsequence $R_{A, B}[8] \circ R_{A, B}[9] \circ R_{A, B}[16] \circ R_{A, B}[20] \circ R_{A, B}[24] = 2 \circ 9 \circ 10 \circ 12 \circ 18$ of $R_{A, B}[8 : 25]$.
}
\label{fig R}
\end{figure}

According to the outline explained above, we reduce the DTW distance problem for the dissimilarity function $d_{A, B}$ to the HIS problem for weighted integer sequence $R_{A, B}$ by appropriately assigning an integer value and occurrence position to each element $q$ in $R_{A, B}$.

For any element $q$ in $R_{A, B}$, let $f(q)$ denote the index such that $R_{A, B}[f(q)] = q$.
For simplicity, we sometimes use $f(i, j)$ to denote $f(r(i, j))$.

As $R_{A, B}$, we adopt the one such that each element appears in the ``row-wise forward-backward-alternating'' order on the grid arrangement shown in Figure~\ref{fig d}(a) and each element is the integer that indicates its rank in the ``column-wise forward-backward-alternating'' order shown in Figure~\ref{fig d}(b).
More formally, we define $f(q)$, $q$, and $w(q)$ for any element $q$ in $R_{A, B}$ as follows (see also Figure~\ref{fig R}(c) for a concrete example).

\begin{definition}\label{def R}
For any index pair $(i, j)$ with $1 \leq i \leq |A|$ and $1 \leq j \leq |B|$, let $f(r(i, j)) = (i - 1)(2|B| - 1) + j$, $r(i, j) = (j - 1)(2|A| - 1) + i$, and $w(r(i, j)) = c - d_{A, B}(i, j)$.
For any index pair $(i, j)$ with $2 \leq i \leq |A|$ and $2 \leq j \leq |B|$, let $f(\tilde{r}(i, j)) = (i - 1)(2|B| - 1) - j + 2$, $\tilde{r}(i, j) = (j - 1)(2|A| - 1) - i + 2$, and $w(\tilde{r}(i, j)) = c$.
\end{definition}

The correctness of our definition of $R_{A, B}$ is guaranteed from the following two lemmas.
The first lemma presents what condition $R_{A, B}$ should satisfy and the second lemma claims that $R_{A, B}$ satisfies the condition.

\begin{lemma}\label{lem condR}
Suppose that a subsequence $q \circ q'$ of $R_{A, B}$ is a maximal increasing subsequence of $R_{A, B}[f(q) : f(q')]$ if and only if all of $\lceil i_q \rceil \leq i_{q'}$, $\lceil j_q \rceil \leq j_{q'}$, and $(i_{q'} - i_q) + (j_{q'} - j_q) = 1$ hold.
Then a subsequence $Q$ of $R_{A, B}$ is an $[r(i_\vdash, j_\vdash) : r(i_\dashv, j_\dashv)]$-banded maximal increasing subsequence of $R_{A, B}[f(i_\vdash, j_\vdash) : f(i_\dashv, j_\dashv)]$ if and only if $Q = \hat{P}$ for some alignment $P$ of $A[i_\vdash : i_\dashv]$ and $B[j_\vdash : j_\dashv]$.
\end{lemma}

\begin{proof}
Since $R_{A, B}[f(i_\vdash, j_\vdash)] = r(i_\vdash, j_\vdash)$ and $R_{A, B}[f(i_\dashv, j_\dashv)] = r(i_\dashv, j_\dashv)$, it is easy to verify that a subsequence $Q$ of $R_{A, B}$ is an $[r(i_\vdash, j_\vdash) : r(i_\dashv, j_\dashv)]$-banded maximal increasing subsequence of $R_{A, B}[f(i_\vdash, j_\vdash) : f(i_\dashv, j_\dashv)]$ if and only if $Q[1] = r(i_\vdash, j_\vdash)$, $Q[|Q|] = r(i_\dashv, j_\dashv)$, and any contiguous subsequence $q \circ q'$ of $Q$ is a maximal increasing subsequence of $R_{A, B}[f(q) : f(q')]$.
It is also easy to verify that $Q = \hat{P}$ for some alignment $P$ of $A[i_\vdash : i_\dashv]$ and $B[j_\vdash : j_\dashv]$ if and only if $Q[1] = r(i_\vdash, j_\vdash)$, $Q[|Q|] = r(i_\dashv, j_\dashv)$, and any contiguous subsequence $q \circ q'$ of $Q$ satisfies that all of $\lceil i_q \rceil \leq i_{q'}$, $\lceil j_q \rceil \leq j_{q'}$, and $(i_{q'} - i_q) + (j_{q'} - j_q) = 1$ hold.
The lemma follows from the above facts.
\end{proof}

\begin{lemma}\label{lem R}
A subsequence $q \circ q'$ of $R_{A, B}$ is a maximal increasing subsequence of $R_{A, B}[f(q) : f(q')]$ if and only if all of $\lceil i_q \rceil \leq i_{q'}$, $\lceil j_q \rceil \leq j_{q'}$, and $(i_{q'} - i_q) + (j_{q'} - j_q) = 1$ hold.
\end{lemma}

\begin{proof}
From Definition~\ref{def R}, it is easy to verify that a subsequence $q \circ q'$ of $R_{A, B}$ is increasing (i.e., both $f(q) < f(q')$ and $q < q'$ hold) if and only if both $\lceil i_q \rceil \leq i_{q'}$ and $\lceil j_q \rceil \leq j_{q'}$ hold.
On the other hand, $(i_{q'} - i_q) + (j_{q'} - j_q) \geq 1$ holds for any distinct elements $q$ and $q'$ in $R_{A, B}$ with $\lceil i_q \rceil \leq i_{q'}$ and $\lceil j_q \rceil \leq j_{q'}$.
Furthermore, if $(i_{q'} - i_q) + (j_{q'} - j_q) > 1$, then there exists at least an element $q''$ in $R_{A, B}$ with $\lceil i_q \rceil \leq i_{q''}$, $\lceil j_q \rceil \leq j_{q''}$, $\lceil i_{q''} \rceil \leq i_{q'}$, and $\lceil j_{q''} \rceil \leq j_{q'}$ such that $(i_{q''} - i_q) + (j_{q''} - j_q) = 1$.
These facts immediately yield the lemma.
\end{proof}

Now we have the following theorem, presenting a reduction of the DTW distance problem to the HIS weight problem.

\begin{theorem}\label{theo R}
The DTW distance between $A[i_\vdash : i_\dashv]$ and $B[j_\vdash : j_\dashv]$ can be calculated as $c(|A[i_\vdash : i_\dashv]| + |B[j_\vdash : j_\dashv]| + 1)$ minus the $[r(i_\vdash, j_\vdash) : r(i_\dashv, j_\dashv)]$-banded HIS weight of $R_{A, B}[f(i_\vdash, j_\vdash) : f(i_\dashv, j_\dashv)]$.
\end{theorem}

\begin{proof}
Recalling that there exists an $[r(i_\vdash, j_\vdash) : r(i_\dashv, j_\dashv)]$-banded HIS of $R_{A, B}[f(i_\vdash, j_\vdash) : f(i_\dashv, j_\dashv)]$ that is maximal, the theorem follows from Lemmas~\ref{lem condR} and \ref{lem R}, together with the outline of our reduction mentioned earlier.
\end{proof}

Since $q < r(1, j_\vdash) \leq r(i_\vdash, j_\vdash)$ for any element $q$ in $R_{A, B}$ with $f(i_\vdash, 1) \leq f(q) < f(i_\vdash, j_\vdash)$, the resulting Theorem~\ref{theo R} after replacing $f(i_\vdash, j_\vdash)$ with $f(i_\vdash, 1)$ still holds.
Similarly, $f(i_\dashv, j_\dashv)$, $r(i_\vdash, j_\vdash)$, and $r(i_\dashv, j_\dashv)$ in the theorem can respectively be replaced with $f(i_\dashv, |B|)$, $r(1, j_\vdash)$, and $r(|A|, j_\dashv)$, to eventually obtain the following corollary.

\begin{corollary}\label{cor R}
The DTW distance between $A[i_\vdash : i_\dashv]$ and $B[j_\vdash : j_\dashv]$ can be calculated as $c(|A[i_\vdash : i_\dashv]| + |B[j_\vdash : j_\dashv]| + 1)$ minus the $[r(1, j_\vdash) : r(|A|, j_\dashv)]$-banded HIS weight of $R_{A, B}[f(i_\vdash, 1) : f(i_\dashv, |B|)]$.
\end{corollary}

Remark: Theorem~\ref{theo R} and Corollary~\ref{cor R} hold also for any dissimilarity function $d_{A, B}$ such that $d_{A, B}(i, j)$ is an arbitrary nonnegative real number, because the condition that $d_{A, B}(i, j)$ is an integer is not used to derive them.

\subsection{Reduction of the HIS problem to the LIS problem}\label{sec S}

\begin{figure}[t]
\centering
\includegraphics[width=12cm]{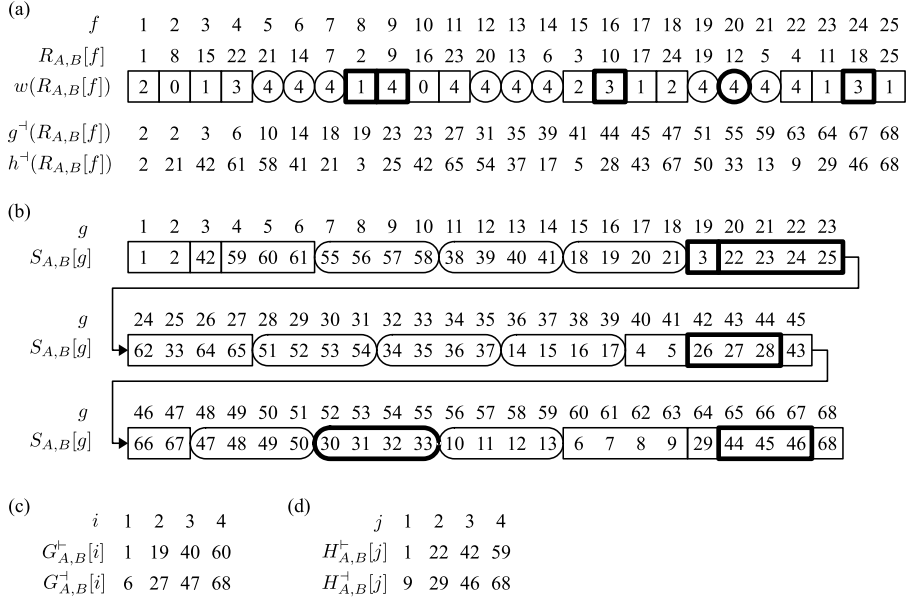}
\caption{
(a) Values $g_\dashv(R_{A, B}[f])$ and $h_\dashv(R_{A, B}[f])$ with $1 \leq f \leq |R_{A, B}|$; (b) Sequence $S_{A, B}$; (c) Arrays $G^\vdash_{A, B}$ and $G^\dashv_{A, B}$; (d) Arrays $H^\vdash_{A, B}$ and $H^\dashv_{A, B}$ for the same dissimilarity function $d_{A, B}$ as Figure~\ref{fig d}, where the alignment $P$ in Figure~\ref{fig d} corresponds to the $[1 : 46]$-banded increasing subsequence $S_{A, B}[19] \circ S_{A, B}[20 : 23] \circ S_{A, B}[42 ; 44] \circ S_{A, B}[52 : 55] \circ S_{A, B}[65 : 67] = 3 \circ 22 \circ 23 \circ 24 \circ 25 \circ 26 \circ 27 \circ 28 \circ 30 \circ 31 \circ 32 \circ 33 \circ 44 \circ 45 \circ 46$ of $S_{A, B}[19 : 68]$.
}
\label{fig S}
\end{figure}

To complete the reduction of the DTW distance problem for the dissimilarity function $d_{A, B}$ to the LIS length problem for the DTW distance sequence $S_{A, b}$, we reduce the HIS problem for $R_{A, B}$ to the LIS problem for $S_{A, B}$.
This is done in a straightforward manner, in which we transform $R_{A, B}$ to $S_{A, B}$ by replacing each element $q$ in $R_{A, B}$ with the sequence $S_q$ of $w(q)$ consecutive integers starting from $1$, if $q = 1 \ (= r(1, 1))$, or just after the last integer of $S_{q - 1}$, otherwise.
Since $R_{A, B}$ consists of $|A||B| + (|A| - 1)(|B| - 1)$ integers each weighted by an integer between $0$ and $c$, the resulting $S_{A, B}$ consists of $O(c|A||B|)$ unweighted integers.
Formally, the DTW distance sequence and its auxiliary arrays, providing the indices corresponding to $f(i, 1)$, $f(i, |B|)$, $r(1, j)$, and $r(|A|, j)$ in Corollary~\ref{cor R}, are defined as follows.
See also Figure~\ref{fig S} for a concrete example.

\begin{definition}\label{def S}
For any integer $q$ in $R_{A, B}$, let $g_\dashv(q)$ (resp. $h_\dashv(q)$) be the sum of $w(q')$ over all elements $q'$ in $R_{A, B}$ such that $f(q') \leq f(q)$ (resp. $q' \leq q$).
Furthermore, let $g_\vdash(q) = g_\dashv(q) - w(q) + 1$, let $h_\vdash(q) = h_\dashv(q) - w(q) + 1$, and let $S_q$ be the sequence $h_\vdash(q) \circ (h_\vdash(q) + 1) \circ \cdots \circ h_\dashv(q)$ of $w(q)$ consecutive integers.
Let the DTW distance sequence $S_{A, B}$ of $A$ and $B$ be the concatenation $S_{R_{A, B}[1]} \circ S_{R_{A, B}[2]} \circ \cdots \circ S_{R_{A, B}[|R_{A, B}|]}$.
As its auxiliary arrays, let $G^\vdash_{A, B}$ (resp. $G^\dashv_{A, B}$) be the array of $|A|$ indices $G^\vdash_{A, B}[i] = g_\vdash(r(i, 1))$ (resp. $G^\dashv_{A, B}[i] = g^\dashv(r(i, |B|))$ ) with $1 \leq i \leq |A|$, and let $H^\vdash_{A, B}$ (resp. $H^\dashv_{A, B}$) be the array of $|B|$ indices such that $H^\vdash_{A, B}[j] = h^\vdash(r(1, j))$ (resp. $H^\dashv_{A, B}[j] = h^\dashv(r(|A|, j))$) with $1 \leq j \leq |B|$.
\end{definition}

\begin{lemma}\label{lem S}
The $[r(1, j_\vdash) : r(|A|, j_\dashv)]$-banded HIS weight of $R_{A, B}[f(i_\vdash, 1) : f(i_\dashv, |B|)]$ is equal to the $[H^\vdash_{A, B}[j_\vdash] : H^\dashv_{A, B}[j_\dashv]]$-banded LIS length of $S_{A, B}[G^\vdash_{A, B}[i_\vdash] : G^\dashv_{A, B}[j^\dashv]]$.
\end{lemma}

\begin{proof}
For any element $s$ in $S_{A, B}$, let $f(s)$ denote the index such that $S_{A, B}[f(s)] = s$.
It follows from Definition~\ref{def S} that an increasing subsequence $s \circ s'$ of $S_{A, B}[f(s) : f(s')]$ is maximal if and only if either 
\begin{itemize}
\item $s \circ s'$ is a contiguous subsequence of $S_q$ for some element $q$ in $R_{A, B}$, or 
\item $s = g_\dashv(q)$ and $s' = g_\vdash(q')$ for some maximal increasing subsequence $q \circ q'$ of $R_{A, B}[f(q) : f(q')]$.
\end{itemize}
From this, it is easy to verify that $T$ is an $[H^\vdash_{A, B}[j_\vdash] : H^\dashv_{A, B}[j_\dashv]]$-banded maximal increasing subsequence of $S_{A, B}[G^\vdash_{A, B}[i_\vdash] : G^\dashv_{A, B}[j^\dashv]]$ if and only if $T = S_{Q[1]} \circ S_{Q[2]} \circ \cdots \circ S_{Q[|Q|]}$ for some $[r(1, j_\vdash) : r(|A|, j_\dashv)]$-banded maximal increasing subsequence $Q$ of $R_{A, B}[f(i_\vdash, 1) : f(i_\dashv, |B|)]$.
Since $|T| = w(Q)$ due to Definition~\ref{def S}, the lemma holds.
\end{proof}

Corollary~\ref{cor R} and Lemma~\ref{lem S} immediately complete our reduction of the DTW distance problem to the LIS length problem as follows.

\begin{theorem}\label{theo S}
The DTW distance between $A[i_\vdash : i_\dashv]$ and $B[j_\vdash : j_\dashv]$ can be calculated as $c(|A[i_\vdash : i_\dashv]| + |B[j_\vdash : j_\dashv]| + 1)$ minus the $[H^\vdash_{A, B}[j_\vdash] : H^\dashv_{A, B}[j_\dashv]]$-banded LIS length of $S_{A, B}[G^\vdash_{A, B}[i_\vdash] : G^\dashv_{A, B}[j^\dashv]]$.
\end{theorem}

It is well known that the dynamic programming (DP) algorithm determines the DTW distance between any subsequences $A[i_\vdash : i_\dashv]$ and $B[j_\vdash : j_\dashv]$ in $O((i_\dashv - i_\vdash)(j_\dashv - j_\vdash))$ time from scratch.
On the other hand, even if $c$ can be treated as a constant, it takes $O((i_\dashv - i_\vdash)(j_\dashv - j_\vdash) \log\log \ell)$ time~\cite{CP} to determine the $[H^\vdash_{A, B}[j_\vdash] : H^\dashv_{A, B}[j_\dashv]]$-banded LIS length of $S_{A, B}[G^\vdash_{A, B}[i_\vdash] : G^\dashv_{A, B}[i_\dashv]]$ from $S_{A, B}$, where $\ell$ is the banded LIS length to be determined.
Consequently, as long as we are in the situation where determining the DTW distance between any given subsequences $A[i_\vdash : i_\dashv]$ and $B[j_\vdash : j_\dashv]$ is required, naively using the DP algorithm is better than maintaining $S_{A, B}$ to apply Theorem~\ref{theo S}.
However, certain kinds of the DTW distance-related problems are relevant to the DTW distances between $A[i_\vdash : i_\dashv]$ and $B[j_\vdash : j_\dashv]$ only for restricted pairs of them, and in such cases, our elaborate representation of the DTW distances by the banded LIS lengths makes sense, as demonstrated in Section~\ref{sec app}.

\section{Applications}\label{sec app}

The reduction of the DTW distance problem to the LIS length problem proposed in Section \ref{sec reduction} becomes meaningful, when we apply the semi-local sequence comparison technique for a pair of sequences, developed by Tiskin~\cite{Tis}.
Here, by semi-local we mean that any pair of an arbitrary prefix of one sequence and an arbitrary suffix of the other or any pair of an arbitrary contiguous subsequence of one and the entire sequence of the other.
This technique was developed so as to be applicable to the longest common subsequence length problem and guarantees, for our particular case considering the banded LIS length, existence of the following useful permutation, which can be constructed efficiently.

\begin{lemma}[\cite{Tis}]\label{lem Tis}
For any pair of time series $A$ and $B$, there exists a permutation sequence $\Pi_{A, B}$ of integers from $1$ to $2|S_{A, B}|$ such that, for any pair of indices $k_\vdash$ and $k_\dashv$ with $1 \leq k_\vdash, k_\dashv \leq 2|S_{A, B}| - 1$, $\min(k_\dashv, |S_{A, B}|) - \max(0, k_\vdash - |S_{A, B}|)$ minus the number of indices $k$ with $k_\vdash + 1 \leq k \leq 2|S_{A, B}|$ and $1 \leq \Pi_{A, B}[k] \leq k_\dashv$ is equal to 
\begin{itemize}
\item the $[1 : k_\dashv]$-banded LIS length of $S_{A, B}[|S_{A, B}| - k_\vdash + 1 : |S_{A, B}|]$, if both $k_\vdash$ and $k_\dashv$ are less than or equal to $|S_{A, B}|$, 
\item the $[k_\vdash - |S_{A, B}| + 1 : |S_{A, B}|]$-banded LIS length of $S_{A, B}[1 : 2|S_{A, B}| - k_\dashv]$, if both $k_\vdash$ and $k_\dashv$ are greater than or equal to $|S_{A, B}|$, 
\item the $[k_\vdash - |S_{A, B}| + 1 : k_\dashv]$-banded LIS length of $S_{A, B}$, if $k_\dashv \leq |S_{A, B}| \leq k_\vdash$ and $k_\vdash - |S_{A, B}| + 1 \leq k_\dashv$, and 
\item the LIS length of $S_{A, B}[|S_{A, B}| - k_\vdash + 1 : 2|S_{A, B}| - k_\dashv]$, if $k_\vdash \leq |S_{A, B}| \leq k_\dashv$ and $|S_{A, B}| - k_\vdash + 1 \leq 2|S_{A, B}| - k_\dashv$.
\end{itemize}
\end{lemma}

\begin{lemma}[\cite{Tis} with any of \cite{Tis2} or \cite{Sak}]\label{lem Tis const}
Sequence $\Pi_{A, B}$ in Lemma \ref{lem Tis} can be obtained from $S_{A, B}$ in $O(|S_{A, B}| \log^2 |S_{A, B}|)$ time and $O(|S_{A, B}|)$ space.
\end{lemma}

Once $\Pi_{A, B}$ is implemented as the two-dimensional range counting tree~\cite{Cha}, in $O(|S_{A, B}| \log |S_{A, B}|)$ time and $O(|S_{A, B}|)$ space, the number of indices $k$ with $k_\vdash + 1 \leq k \leq 2|S_{A, B}|$ and $1 \leq \Pi_{A, B}[k] \leq k_\dashv$ for any such pair of indices $k_\vdash$ and $k_\dashv$ can be determined in $O(\log |S_{A, B}|)$ time.

The following DTW distance-related problems are included in typical kinds of problems efficiently handled by the semi-local sequence comparison technique.
In what follows, we assume that any dissimilarity function takes as its value one of integers from $0$ to $c$ with $c = O(1)$, independent of the length of the target pair of time series.
Compared to a naive use of the DP algorithm, our reduction to the LIS problem allows us to solve the problems asymptotically faster by an almost linear factor.
The drawback of our reduction is its space-inefficiency.
The DP algorithm (with Hirschberg's divide-and-conquer technique~\cite{Hirschberg}) requires only linear space, while ours consumes quadratic space.
As a result, the algorithms we will propose for the problems based on our reduction technique balance execution speed and space consumption almost equally.

\subsection{The circular DTW distance problem}\label{sec cDTWp}

Given a pair of time series $A$ and $B$ with $|A| \leq |B|$, the circular DTW distance problem consists of determining the minimum of the DTW distance between $A'' \circ A'$ and $B$ over all partitions of $A$ into a prefix $A'$ and the remaining suffix $A''$, together with an arbitrary circular shift $A'' \circ A'$ of $A$ that achieves this minimum distance with $B$.
This problem may arise, for example, when we have a pair of daily temperature data for a year taken at different locations or environments and want to know the similarity and phase shift between them.

A naive algorithm solves the problem in $O(|A|^2 |B|)$ time and $O(|A|)$ space by determining the DTW distance between $A'' \circ A'$ and $B$ using the DP algorithm in $O(|A||B|)$ time for each partition of $A$ into $A' \circ A''$ and taking the minimum.
In contrast, if the two-dimensional range counting tree $T_{A \circ A, B}$ for $\Pi_{A \circ A, B}$ is available, then the problem can be solved in $O(|A| \log |B|)$ time by determining the DTW distance between $(A \circ A)[i : i + |A| - 1] \ (= A[i: |A|] \circ A[1 : i - 1])$ and $B$ in $O(\log |\Pi_{A \circ A, B}|)$ time for each index $i$ from $1$ to $|A|$ and taking the minimum.
Furthermore, $T_{A \circ A, B}$ can be constructed from scratch in $O(|A||B| \log^2 |B|)$ time.
Consequently, the following holds.

\begin{theorem}\label{theo cDTWp}
Given a pair of time series $A$ and $B$, the circular DTW distance problem can be solved in $O(|A||B| \log^2 |B|)$ time and $O(|A||B|)$ space.
\end{theorem}

\subsection{The square root DTW distance problem}\label{sec sqDTWp}

Given a time series $A$, the square root DTW distance problem is to find an arbitrary partition of $A$ into a prefix $A[1 : i]$ and the remaining suffix $A[i + 1, |A|]$ that minimizes the DTW distance between them and to determine this DTW distance.
This problem may arise, for example, when we want to test if a time series can be thought of as the concatenation of a pair of inexact copies of an unknown pattern.

Similarly to the case of the circular DTW distance problem, a naive algorithm can solve this problem in $O(|A|^3)$ time and $O(|A|)$ space by determining the DTW distance between $A'$ and $A''$ in $O(|A|^2)$ time based on the DP algorithm for each partition of $A$ into $A' \circ A''$ and taking the minimum.
In contrast, if the two-dimensional range counting tree $T_{A, A}$ for $\Pi_{A, A}$ is available, then the problem can be solved in $O(|A| \log |A|)$ time by determining the DTW distance between $A[1 : i]$ and $A[i + 1 : |A|]$ in $O(\log |\Pi_{A, A}|)$ time for each index $i$ from $1$ to $|A| - 1$ and taking the minimum.
Furthermore, $T_{A, A}$ can be constructed from scratch in $O(|A|^2 \log^2 |A|)$ time.

\begin{theorem}\label{theo sqDTWp}
Given a time series $A$, the square root DTW distance problem can be solved in $O(|A|^2 \log^2 |A|)$ time and $O(|A|^2)$ space.
\end{theorem}

\subsection{The periodic DTW distance problem}\label{seq pDTWp}

Given a pair of time series $A$ and $B$ with $|A| \leq |B|$, the periodic DTW distance problem is to find an arbitrary quadruple consisting of a nonnegative integer $\ell$, a decomposition $B_0 \circ B_1 \circ \cdots \circ B_\ell$ of $B$, and indices $i_\vdash$ and $i_\dashv$ with $1 \leq i_\vdash, i_\dashv \leq |A|$ that minimizes $\mathrm{DTW}(A[i_\vdash : i_\dashv], B)$, if $\ell = 0$, or $\mathrm{DTW}(A[i_\vdash : |A|], B_0) + \sum_{k = 1}^{\ell - 1} \mathrm{DTW}(A, B_k) + \mathrm{DTW}(A[1 : i_\dashv], B_\ell)$, otherwise.
Here, $\mathrm{DTW}(A', B')$ denotes the DTW distance between $A'$ and $B'$.
This problem may arise, for example, when we have a time series that can be thought of as from an inexact tandem repeat of a known specific pattern, and want to cut it into the inexact occurrences of the pattern.

If the two-dimensional range counting tree $T_{A, B}$ for $\Pi_{A, B}$ is available, then the problem can be solved in $O(\max(|A| \log |B|, |B|)|B| \log |B|)$ time as follows.
Let $G$ be the directed acyclic graph consisting of source vertices $u_i$ and sink vertices $v_i$ both with $1 \leq i \leq |A|$ and internal vertices $w_j$, each also denoted by $x_{j + 1}$, with $1 \leq j \leq |B| - 1$,  
\begin{itemize}
\item edges each from a source vertex $u_{i_\vdash}$ to a sink vertex $v_{i_\dashv}$ with $i_\vdash \leq i_\dashv$, the weight of which is set to the DTW distance between $A[i_\vdash : i_\dashv]$ and $B$, 
\item edges each from a source vertex $u_i$ to an internal vertex $w_j$, the weight of which is set to the DTW distance between $A[i : |A|]$ and $B[1 : j]$, 
\item edges from an internal vertex $x_{j_\vdash}$ to another internal vertex $w_{j_\dashv}$ with $j_\vdash \leq j_\dashv$, the weight of which is set to the DTW distance between $A$ and $B[j_\vdash : j_\dashv]$, and 
\item edges each from an internal vertex $x_j$ to a sink vertex $v_i$, the weight of which is set to the DTW distance between $A[1 : i]$ and $B[j : |B|]$.
\end{itemize}
Hence, each path on $G$ from a source vertex to a sink vertex corresponds to a distinct quadruple of $\ell$, $B_0 \circ B_1 \circ \cdots \circ B_\ell$, $i_\vdash$, and $i_\dashv$, and vice versa.
Furthermore, the weight of the path is equal to $\mathrm{DTW}(A[i_\vdash : i_\dashv], B)$, if $\ell = 0$, or $\mathrm{DTW}(A[i_\vdash : |A|], B_0) + \sum_{k = 1}^{\ell - 1} \mathrm{DTW}(A, B_k) + \mathrm{DTW}(A[1 : i_\dashv], B_\ell)$, otherwise, where the weight of a path on $G$ is the sum of the weights of all edges in the path.
This implies that the periodic DTW distance problem can be solved by finding an arbitrary path on $G$ from a source vertex to a sink vertex that has minimum weight.
If $G$ is available, then such a path can be found in time linear in the number of edges in $G$, which is $O(|B|^2)$, and in space linear in the number of vertices in $G$, which is $O(|B|)$, by determining the midpoint of the path recursively in a straightforward way.
Instead of constructing $G$ explicitly, we can use the two-dimensional range counting tree $T_{A, B}$ for $S_{A, B}$ as a data structure that supports $O(\log |B|)$-time queries of the weight of any edge in $G$, which allows us to obtain the path in $O(|B|^2 \log |B|)$ time and $O(|B|)$ space, excluding space for storing $T_{A, B}$.
Furthermore, $T_{A, B}$ can be constructed from scratch in $O(|A| |B| \log^2 |B|)$ time and $O(|A| |B|)$ space.
(Adopting the same strategy, we can design an $O(|A| |B|^2)$-time, $O(|B|)$-space algorithm based on DP.)

\begin{theorem}\label{theo pDTWp}
Given a pair of time series $A$ and $B$ with $|A| \leq |B|$, the periodic DTW distance problem can be solved in $O(\max(|A| \log |B|, |B|)|B| \log |B|)$ time and $O(|A| |B|)$ space.
\end{theorem}

\section{Concluding remarks}\label{sec conc}

This article showed that for any pair of time series $A$ and $B$ and any dissimilar function mapping any pair of elements of $A$ and $B$ to an integer in $\{ 0,1,\dots,c \}$, there exists a sequence $S$ of $O(c|A||B|)$ integers such that the DTW distance between any contiguous subsequence of $A$ and any contiguous subsequence of $B$ can be represented by the banded LIS length of a contiguous subsequence of $S$.
As applications of this reduction of DTW to LIS, novel algorithms for three DTW-related problems, the circular, square root, and periodic DTW distance problems, were presented utilizing the semi-local sequence comparison technique of Tiskin~\cite{Tis} originally developed for LCS-related problems.

Compared with the naive DP-based algorithms for the DTW-related problems, the proposed algorithms run asymptotically faster but consume more space.
An immediate question from this time-space trade-off is whether space-inefficiency of our algorithms can be removed by reducing the required space from quadratic to linear.
Another question also comes from the quadratic length of the integer sequence representing the DTW distance by its LIS length.
Due to this length, there is a gap between the size of the permutation sequence used to solve the semi-local LCS and DTW problems: linear for LCS and quadratic for DTW.
The fully-local LCS problem, answering queries of an LCS between any given pair of contiguous subsequences, has an interesting trade-off between space consumption and query times.
That is, the DP algorithm finds an LCS from scratch in quadratic time using linear space, while a quadratic-time constructible data structure can support linear-time queries of an LCS~\cite{Sak2}.
Can we have the same trade-off also on the fully-local DTW problem?
In other words, are there any quadratic-space (or even quadratic-time constructible) data structures supporting linear-time queries of a DTW alignment between any pair of contiguous subsequences?
All the aforementioned questions could be resolved if one can find a way to apply the semi-local LCS comparisons of Tiskin~\cite{Tis} more directly to the case of DTW, without using a reduction to LIS. So far we have not been able to find such a method.

\section*{Acknowledgements}

\noindent The work of Shunsuke Inenaga was supported by JST PRESTO Grant Number JPMJPR1922.

\bibliographystyle{spmpsci}      

\end{document}